\documentclass[11pt]{article}
\usepackage[centertags]{amsmath}
\usepackage[mathscr]{eucal} 
\usepackage{amsfonts}
\usepackage{amssymb}
\usepackage{amsthm}
\usepackage[utf8]{inputenc}
\usepackage{amsmath}
\usepackage{amsfonts}
\usepackage{amssymb}
\usepackage{graphicx}
\usepackage{slashbox}

\usepackage[numbers]{natbib}
\usepackage{etoolbox}
\makeatletter
\apptocmd{\thebibliography}{\global\c@NAT@ctr 0\relax}{}{}
\makeatother
   \renewcommand{\theequation}{\thesection.\arabic{equation}}
   \csname @addtoreset\endcsname{equation}{section}

\newtheorem{theorem}{Theorem}[section]
\newtheorem{lemma}{Lemma}[section]
\newtheorem{assumption}{Assumption}[section]

\newtheorem{corollary}{Corollary}[section]

\newtheorem{definition}{Definition}[section]

\def\psii{U}
\def\UU{\eta}
\def\UUU{\zeta}
\date{Submitted: March 18, 2014. Revised: November 25, 2014 }
\begin{document}
\title{
On asymptotic optimality of Merton's myopic
portfolio strategies for discrete time market  
}
\author{Alexandra Rodkina
\footnote{ Department of Mathematics,  The University of the West
Indies, Kingston 7, Jamaica }
and Nikolai Dokuchaev\footnote{Department
of Mathematics and Statistics, Curtin University, Australia}\footnote{Corresponding author}}


\maketitle

\begin{abstract}
This paper studies the properties of discrete time stochastic optimal
control problems associated with portfolio selection.
We investigate if  optimal  continuous time strategies  can be
  used effectively for a discrete time market after a straightforward discretization.
  We found  that Merton's  strategy approximates
  the performance of the optimal strategy in a discrete
time model with the  sufficiently small time steps.
\end{abstract}

\def\uu{u}

\section{Introduction}  The paper studies discrete time stochastic optimal control problems
and their relationships with continuous time optimal control
problems. More precisely, we study optimal investment problems where
$\mathbf E U(X_T)$ is to be maximized. Here $X_T$ represents the
total wealth at final time $T$, and $U(\cdot)$ is a utility
function. We consider the case where $U(x)=x^\alpha$, $\alpha\in
(0,1)$. For continuous time market models, these utilities have a
special significance, in particular, because the optimal strategies
for them  are known explicitly  (so-called Merton's strategies).
These strategies are {\it myopic}; they depend only on the current
observations of the market parameters, including the risk free rate,
the appreciation rate, and the volatility matrix, even for the case
of unknown prior distributions and evolution law. The optimality of
Merton's strategies still holds when the market parameters are
random and independent of the driving Brownian motion, i.e., in the
case of "totally unhedgeable" coefficients, according to \citet{KS},
Chapter~6. The solution that leads to myopic strategies goes back to
\citet{Merton}; the case of random coefficients was discussed in
\citet{KS} and \citet{DH}. \index{Dokuchaev and Haussmann (2001).}
\par
 The real stock
prices are presented as time series, so the discrete time
(multi-period) models are more natural than continuous time models.
On the other hand, continuous-time models give a good description of
distributions and  often allow explicit solutions of optimal
investment problems.
  \par
   For a real market,  a formula for an optimal strategy derived for a continuous-time model can often be
effectively used after the natural discretization. However, this
strategy will not be  optimal for time series observed in the real
market. Therefore, it is important to extend the class of discrete
time models that allow myopic and explicit optimal portfolio
strategies. The problem of discrete-time portfolio selection has
been studied in the literature, such as in  \citet{Sm}, \citet{Le},
\citet{Mo,Merton,Sa,Fa,Haa,Hab,EG,Fr,DL,Os,GH,Pl,LN,Xu,CO,Z}. \index{\citet{Sm}, \citet{Le},
\citet{Mo}, \citet{Merton}, \citet{Sa}, \citet{Fa}, \citet{Haa},
\citet{Hab}, \citet{EG}, \citet{Fr}, \citet{DL}, \citet{Os},
\citet{GH}, \citet{Pl}, \citet{LN}, \citet{Xu}, \citet{CO},
\citet{Z}.}  If a discrete time
market model is complete, then the martingale method can be used
(see, e.g., \citet{Pl}). Unfortunately, a discrete time market model
can be complete only under very restrictive assumptions. For
incomplete discrete time markets, the main tool is dynamic
programming that requires to derive and solve a backward Bellman equation with a Cauchy condition at the
terminal time. For the general case, this procedure  involves  recalculation of  the conditional densities
at each time step which is numerically challenging (see, e.g.,
\citet{Pl} or  \citet{GS}). This is why the optimal investment
problems for discrete time can be more difficult than for the continuous
time setting where explicit solutions are often possible.
\par
There are several special cases when an investment problem allows for an
explicit solution in discrete time, and, for some cases, optimal
strategies are myopic (see \citet{Le}, \citet{Mo}, \citet{Haa}).
However, the optimal strategy is not myopic and it cannot be
presented explicitly for power utilities in the general case.
\citet{Haa}  showed that the optimal strategy is not myopic for
$U(x)=\sqrt{x}$ if returns have serial correlation and evolve as a
Markov process.
 \par
 In a mean-variance discrete time multi-period
setting, the optimal strategies represent some analog of Merton's
strategies. These strategies are myopic for mean-variance goal
achieving problems and non-myopic if these goals have to be selected
to solve a problem with constraints; see  \citet{LN}, \citet{D10},
\citet{Z}. It appears that the problems with utility functions
$\psii(x)=x^\alpha$, $\alpha<1$, have different properties with
respect to time disretization. In particular, \citet{D07}
demonstrated that the direct discretization of continuous time
optimal Merton's strategies does not approximate the optimal
strategy for the discrete time market for concave utility functions
$\psii(x)=x^\alpha$ such that $\psii(x)=-\infty$ for $x<0$. More
precisely,
 the difference between the optimal expected utilities for discrete
time and continuous time models does not disappear if the number of
periods (or frequency of adjustments) grows. As the result, the
optimal expected utility calculated for a continuous time market
cannot be approximated by piecewise constant strategies with
possible jumps at given times $(t_k)_{k=1}^N$, even if $N\to
+\infty$ and $t_k-t_{k+1}\to 0$.
 \par
In the present paper, we  reconsider  the discrete time  optimal
portfolio selection problems. We suggest a solution based on myopic
Merton's strategies that are optimal for   related continuous time
portfolio selection problems.  We investigate the limit properties
of the discrete time optimal portfolio selection problem when the
step of the discretization converges to zero.  We found that the
performance of the discrete time strategy obtained directly from
Merton's strategy approximates  the optimal strategy. This
suboptimal discrete time strategy is myopic.  We consider the case
of Gaussian noise in the discrete time equation for the price. This
means that the stock price can be negative with non-zero
probability; this feature does not affect the validity of the model
since this probability converges to zero as the step of
discretization converges to zero; see \citet{AGRK}.  The proof is
based  on the application of the variant of the discrete It\^o
formula first introduced by \citet{ABR1}. \index{Appleby, Berkolaiko
\& Rodkina (2009).} It can be noted that the proof does not use the
dynamic programming principle.
\par
  These results lead to a conclusion that Merton's strategies can be
  used effectively for the discrete time multi-period market models with power
  utilities $U(x)=x^\alpha$, $\alpha<1$ that have sufficiently small time steps
  and approximate the continuous time model. This seems to  contradict to the result from
  \citet{D07}. However, there is not a contradiction. In the present paper, we
assumed that  $\psii(x)=Lx$ for $x<0$, where  $L>0$ can be selected
to be arbitrarily large. On the other hand, \citet{D07}  assumed
that $\psii(x)=-\infty$ for $x<0$. This difference in the problem
setting leads to different conclusions. Note that the utility
function considered in the present paper is not concave; however,
its shape is becoming  "more concave" as $L\to +\infty$. Moreover,
the impact  of non-concavity of $U$ for any given $L$ disappears
since this probability converges to zero as the step of
discretization converges to zero. We illustrate this in numerical experiments to demonstrate the impact  of the size of the
interval of discretization on
the performance of Merton's strategy and the impact  of the selection  of finite $L$ in the adjusted utility function.
In these experiments, we found that weekly portfolio adjustments is sufficient to compensate the discretization error for Merton's strategy.
Moreover, we found that this error is almost negligible for  the model with daily portfolio adjustments.


\section{ Problem setting}
\label{sec:int}
In this paper we consider the following controlled stochastic difference equation
\begin{equation}
\label{eq:discrw} \begin{split}
&x_{n+1}=x_n\left(1+hu_na_n+\sqrt{h}u_nb_n\xi_{n+1}\right), \quad n=0, 1, \dots, N-1,\\
 &x(0)=x_0>0, \end{split}
\end{equation}
where  $x_0>0$ is nonrandom, $\xi_n$ are random variables, $a_n$,
$b_n$ are nonrandom coefficients,  $u_n$ is a nonrandom control
(strategy), $n=0, 1, \dots, N-1$, $N\in \mathbf N$, $h>0$ is a small
parameter, calculated as
\begin{equation}
\label{def:h}
h:=\frac TN.
\end{equation}
The value $T>0$  is fixed throughout all paper, but $N$ can increase, which makes $h$ decrease.

We can  either  consider equation \eqref{eq:discrw} independently,
or  think about it  as the Eulier-Maruyama discretization of the
following It\^o stochastic equation
\begin{equation}
\label{eq:contw}
dX_t=X_tu(t)\bigl(a(t)dt+b(t)dW_t\bigr), \quad t\in [0, T], \quad X(0)=x_0,
\end{equation}
where $W$ is a standard  Wiener process,  $b, a, u: [0, T]\to
\mathbf R$ are continuous nonrandom functions.  In this setting  $h$
is a step size of discretization of the interval $[0, T]$ and $N$ is
a number of corresponding mesh points.

We recall that the Euler-Maruyama numerical method for equation
\eqref{eq:contw}  computes approximations $x_n(h)\approx X_{nh}$ by
\begin{equation}
\label{EM_def}
x_{n+1}(h)=x_n(h)\bigl(1 + hu(nh)a(nh) + u(nh)b(nh)\Delta W_{n+1}\bigr),
\end{equation}
where $h>0$ is the constant step size and $\Delta W_{n+1}=W((n+1)h)-W(nh)$. We see that when
\[
\xi_{n+1}=\frac{W((n+1)h)-W(nh)}{\sqrt {h}}, \quad a_n=a(nh), \quad
b_n=b(nh), \quad  u_n=u(nh), \quad\] \eqref{EM_def} coincides with
\eqref{eq:discrw} and $\xi_{n+1}$ is a standardized normal random
variable. More information  about Euler-Maruyama discretization and
stochastic difference equations could be found, e.g.,  in
\citet{HMS, KP,  AGRK, ABR1}.

Assume that the following assumptions hold.
\begin{assumption}
\label{as:ab}
Sequences \, $(a_n)_{n\in \mathbf N}$  \, and \, $(b_n)_{n\in \mathbf N}$\,\,
are  nonrandom  and such that for some \\$\hat a$, $\hat b$, $\underline a$, $\underline b>0$, \, $\underline a\le \hat a$, \, $\underline b\le \hat b$,
\begin{equation}
\label{bounds:ab}
\underline a\le|a_n|\le \hat a, \quad\underline b\le | b_n|\le \hat b, \quad \forall n\in \mathbf N.
\end{equation}
\end{assumption}

\begin{assumption}
\label{as:xi}
$(\xi_n)_{n\in \mathbf N}$ is a sequence of  independent and $\mathcal N(0, 1)$ distributed random variables.
\end{assumption}
Let $(\Omega, {\mathcal{F}}, \{{\mathcal{F}}_n\}_{n \in \mathbf N},
{\mathbb P})$ be a complete filtered probability space. We assume that the filtration
$\{{\mathcal{F}}_n\}_{n  \in \mathbf N}$ is naturally generated: $\mathcal{F}_{n+1} = \sigma \{\xi_{i+1} :
i=0,1,...,n\}$, where  the  sequence $(\xi_{n})_{n\in \mathbf N}$ satisfies Assumption \ref{as:xi}.

We use the standard abbreviation ``{\it a.s.}'' for the wordings
``almost sure'' or ``almost surely'' throughout the text.

Among all the sequences $(x_n)_{n \in \mathbf N}$ of the random
variables we distinguish those for which $x_n$ are
${\mathcal{F}}_n$-measurable for all $n \in \mathbf N$.
A detailed exposition of the definitions and facts of the theory of random processes can be
found, e.g., in \citet{Shiryaev96}.

Define  for some $\alpha\in (0,1 )$ and $L>0$,
\begin{equation}
\label{def:psi}
\psii(x)=x^\alpha, \,\,  x\ge 0, \quad \psii(x)=Lx, \,\, x< 0.
\end{equation}

\begin{definition}
\label{def:adstrat}
For a given $N\in\mathbf{N}$, the set $ \mathcal {U}=\mathcal {U}(N)$ of  admissible strategies is the set of all  nonrandom   vectors $\uu =(u_n)_{n=0}^{N-1}$  such that
 \begin{equation}
\label{bound:pi}
\underline u \le |u_n|\le \hat u, \quad n=0,1,...,N-1,
\end{equation}
for some positive numbers $\underline u$, $\hat u$, \, $0< \underline {u}\le \hat u$.
\end{definition}
Up to the end of the paper, we  consider the following optimal control
problem:
 \begin{equation}
\label{def:opt0}
\hbox{Maximize}\quad
\mathbf E\left[\psii(x_N)\right]\quad \hbox{over}\quad u\in\mathcal {U},
\end{equation}
where $x$ is a solution to \eqref{eq:discrw} with $h=\frac TN$ and admissible strategy $\uu $, \, $\mathcal {U }=\mathcal {U }(N)$ is  the set of all admissible strategies  introduced in Definition  \ref{def:adstrat}.

\section{Optimal portfolio selection and the main result}
Problem  \eqref{def:opt0} has applications for optimal portfolio selection.   It appears that
 (\ref{eq:discrw}) describes the dynamic of the total wealth $x_n$ of an investor at time period $n$ for a single stock discrete time market model
 with a risk-free investment with zero return. The dynamic of the stock price is described by the equation \begin{equation}
\label{stock}
s_{n+1}=s_n\left(1+ha_n+\sqrt{h}b_n\xi_{n+1}\right), \quad n=0, 1, \dots, N-1, \quad s_0=1.
\end{equation}
It is assumed that the portfolio is distributed among the
shares of this stock and the risk-free investment with zero return.
A strategy $\uu $ represents a dynamically selected ratio of investment in stock.
More precisely, let  $\gamma_n$ be the quantity of stock shares in the portfolio at time $n$,
then $u_n=\gamma_n s_n/x_n$, where $\gamma_ns_n$ is the current
value of the stock portfolio, $x_n$ is the current total value of the portfolio.
We select the strategy $\uu$  in the class of admissible processes described above
and calculate the quantity of shares $\gamma_n=u_nx_n/s_n$; effectively, we  select  {\em closed-loop} strategies.
It is assumed that the strategy is {\em self-financing}, i.e.,
$$
x_{n+1}-x_n=\gamma_n(s_{n+1}-s_n),\quad n=0,1,2,...,
$$
where $\gamma_n=u_nx_n/s_n$ is the quantity of stock shares in the portfolio at time $n$.
This assumptions means that the model does not include an external sources of funds  and that there is no expenses, transaction costs,  and dividend payments. The increments
of the wealth are defined solely by the stock price changes and by the quantity of the shares.

In fact, the case of non-zero return for the risk free asset is also covered by this model, if one interprets $x_n$ as the discounted wealth  and $s_n$ as the discounted stock price (discounted with respect to the risk-free asset).   A more detailed market model
description can be found, e.g., in \citet{Pl, D07}.

For this discrete time market model,  a standard
problem of optimal portfolio selection is to maximize the expectation of the utility function ${\bf U}(x_N)$ of  the terminal wealth $x_N$, i.e.,  to find  a  strategy ${\uu ^*}$ which solves optimal control problem
 \begin{equation}
\label{sU}
\hbox{Maximize}\quad
\mathbf E\left[{\bf U}(x_N)\right]\quad \hbox{over}\quad  \mathcal {U},
\end{equation}
where ${\bf U}$ is some given concave utility function,
  $x$ is a solution to \eqref{eq:discrw} with $h=\frac TN$, $ \mathcal {U}$ is a set of all admissible strategies  according to Definition  \ref{def:adstrat}.

Further, It\^o equation (\ref{eq:contw}) describes the evolution of the total wealth $X_t$
for a single stock continuous market model with zero risk-free interest rate where the stock price evolution is described by the It\^o equation
\begin{equation}
\label{eq:conts} \begin{split}
&dS_t=S_t\bigl(a(t)dt+b(t)dW_t\bigr), \quad t\in [0, T], \quad S_0=1.
\end{split}
\end{equation}
For this continuous time market model,  a standard optimal portfolio selection problem is to maximize the expectation of the utility function ${\bf U}(X_T)$ of  the terminal wealth $X_T$, i.e.,  to find  a strategy  $u:[0,T]\times\Omega\to{\bf R}$ in a certain class of admissible  strategies $\overline{\mathcal U}$ that solves optimal control problem
 \begin{equation}
\label{sUc}
\hbox{Maximize}\quad
\mathbf E\left[{\bf U}(X_T)\right]\quad \hbox{over}\quad \overline{\mathcal U},
\end{equation}
where  $X_t$ is a solution to \eqref{eq:contw}.
For the case when ${\bf U}(x)=x^\alpha$, $\alpha\in (0,1)$, the following so-called Merton strategy
\begin{equation}
\label{Mer}
u^*(t)=\frac{a(t)}{(1-\alpha)b^2(t)}, \quad t\in[0,T],
\end{equation}
is optimal in the continuous time setting (\ref{sUc}) in the class of
admissible strategies that include  all bounded random processes
adapted to the filtration generated by $S_t$; see, e.g., \citet{KS},
Chapter 6, and \citet{Merton}. In fact, this strategy is optimal in
a even wider class of random and adapted processes $u(t)$, as well as
in the setting with random $a(t)$ and $b(t)$ that are independent from
$W_t$.

It can be seen that problem (\ref{def:opt0}) is in fact a  modification of problem (\ref{sU}).
Note that the "utility function" $U(x)$ in (\ref{def:opt0})
is not concave in $x\in R$; however, its shape is becoming  "more concave" as $L\to +\infty$. \index{Moreover,
we will show that the impact  of non-concavity of $U$ for any given $L$ disappears
since the probability that the wealth achieves zero vanishes as $h\to 0$, for a certain choice of the strategy.}

Consider the strategy  ${\uu ^*}$  such that
\begin{equation}
\label{def:pin}
u^*_n=\frac{a_n}{(1-\alpha)b^2_n}, \quad n=0, 1, \dots, N-1.
\end{equation}
Condition  \eqref{bound:pi} is satisfied for this strategy.  Notice that this strategy represents a direct discretization of Merton's strategy (\ref{Mer}).
It can be also noted that strategy (\ref{def:pin})  does not depend on the choice of $L$.

Our main result can be formulated as the following.
\begin{theorem}\label{ThM}
The  strategy $u^*$ defined by \eqref{def:pin}  maximizes $\mathbf E \psii(x_N)$ approximately for small enough $h=\frac TN$, meaning that
\[
 \sup_{u} \mathbf E \psii(x_N) =\mathbf E \psii(x^*_N)+O(h)\quad \hbox{as}\quad h\to 0,
\]
where $x^*_N$ is the terminal wealth for strategy \eqref{def:pin} and $O(h)\to 0$ as $h\to 0$, independently on $N$.
\end{theorem}

We show that the error of this approximation tends to zero as
step size of discretization  $h\to 0$ (which is equivalent that number of mesh points $N\to \infty$).
The proof is heavily dependent  on the application of the variant of the discrete It\^o formula first introduced in \citet{ABR1},
as well as on the fact that  solution $x^*_n$ of \eqref{eq:discrw} for strategy \eqref{def:pin} is positive for all $n=1, \dots, N$ with probability which tends to one when
$h\to 0$ (or $N\to \infty$); see \citet{AGRK}.

In \citet{D07}, it was shown  that the direct discretization of
continuous time optimal Merton's strategies does not approximate the
optimal strategy for the discrete time market if the utility
function
 $\psii(x)=x^\alpha$ is extended as $\psii(x)=-\infty$ for
$x<0$.  We found that this can be overcome
using the functions $U$ with non-concavity that can be made arbitrarily small with selection of a large $L>0$.
Moreover, we show that the probability that this non-concavity will ever have any impact
vanishes as $h\to 0$, since
the probability that the wealth ever
achieves zero vanishes as $h\to 0$.

Let us review the main steps of the proofs.

Let
\begin{equation}
\label{def:phi}
\phi(x)=|x|^\alpha, \quad x\in \mathbf R.
\end{equation} First, we observe that the solution $x_n$ of \eqref{eq:discrw}  can be represented as
\begin{equation}
\label{pres:0}
x_n=x_0\prod_{i=0}^{n-1}\left(1+hu_ia_i+\sqrt{h}u_ib_i\xi_{i+1}\right), \quad x_0>0, \quad n=1, \dots, N.
\end{equation}
Hence
\begin{equation*}
\mathbf E \phi(x_n)=\phi(x_0)\prod_{i=0}^{n-1}\mathbf E \phi(1+hu_ia_i+\sqrt{h}u_ib_i\xi_{i+1}),\quad n=1, \dots, N.
\end{equation*}
Application of the discrete It\^o formula to each $\mathbf E \phi(1+hu_ia_i+\sqrt{h}u_ib_i\xi_{i+1})$ gives that
\[
\sup_u \mathbf E \phi(x_N)=x_0^\alpha\prod_{n=0}^{N-1}\left[1+\alpha h\frac{a^2_n}{2(1-\alpha)b^2_n}\right]+O(h)\quad \hbox{as}\quad h\to 0.
\]
Then we show that the probability
 \[
\mathbf P\{\omega:  \psii(x_N(\omega))\neq \phi(x_N(\omega))\}
\]
can be made arbitrary  small when $N=T/h$ is big enough.
Finally we prove that  $$
\sup_u \mathbf E\left[\psii(x_N)\right]= x_0^\alpha\prod_{n=0}^{N-1}\left[1+\alpha h\frac{a^2_n}{2(1-\alpha)b^2_n}\right]+O(h)\quad \hbox{as}\quad h\to 0.
$$

The remaining part of the paper is devoted to the proof of Theorem \ref{ThM} accordingly to the outline given above.

\renewcommand{\theequation}{A.\arabic{equation}}
\csname @addtoreset\endcsname{equation}{section}
\renewcommand{\thelemma}{A.\arabic{lemma}}
\csname @addtoreset\endcsname{equation}{section}
\renewcommand{\theproposition}{A.\arabic{proposition}}
\renewcommand{\thetheorem}{A.\arabic{theorem}}
\renewcommand{\thesubsection}{A.\arabic{subsection}}
\csname @addtoreset\endcsname{subsection}{section}
\setcounter{section}{1}
\setcounter{equation}{0}
\setcounter{theorem}{0}
\setcounter{lemma}{0}
\setcounter{proposition}{0}
\section*{Appendix: proofs}
\subsection{Discrete It\^o formula.}
\label{sec:ito}

The Discrete It\^o formula  which we use in this paper is similar to the formula first introduced by Appleby, Berkolaiko \& Rodkina \cite{ABR1} for the proof of stability results for scalar stochastic difference equations. The main purpose of  this formula  is to mimic the classical It\^o formula for continues processes when we deal with the discrete process described by the equation with small parameter $h$,  similar \eqref{eq:discrw}.  \citet{BBKR} demonstrates the use of a  discrete It\^o formula in the context of stochastic numerical analysis.
Theorem \ref{thm:DIto} below deals with the case which is  slightly different than the one considered in \citet{ABR1} and \citet{BBKR}.
Theorem \ref{thm:DIto} can also be obtained as a partial case of Lemma 5.1 from \citet{AD}, where the  It\^o formula was proved for the diagonal system of stochastic difference equations. However it is much easier  to give a sketch of the proof  here than to adapt Lemma 5.1 for \eqref{eq:discrw}.

\begin{theorem}
\label{thm:DIto}
Let Assumptions \ref{as:ab}, \ref{as:xi}  and condition \eqref{bound:pi} hold.
Consider $\phi:\mathbf R\to \mathbf R$ such that there exists $\delta\in (0,1)$ and
$\phi_\delta:\mathbf R\to \mathbf R$ saisfying
\begin{enumerate}
\item[(i)] $\phi_\delta$  has a  bounded third derivative on $\mathbf R$,
\item[(ii)] $\phi_\delta(s)=\phi(s)$  for $s\notin (-\delta, \delta)$,
\item[(iii)]$ |\phi_\delta(s)-\phi(s)|<K$ for some $K>0$ and all $s\in (-\delta, \delta)$.
\end{enumerate}
Then there exists $h_0$ such that, for all $h\le h_0$,   $N\ge \frac{T}{h_0}$, and  $n=0, 1,
\dots, N-1$
\begin{equation}
\label{cal:DIF}
\begin{split}
&\mathbf E\left(\phi(1+hu_na_n+\sqrt{h}u_nb_n\xi_{n+1})\right)=\phi(1)+ h\phi'(1)u_na_n+h\frac{\phi''(1)}2u^2_nb^2_n+o(h),
\end{split}
\end{equation}
where
\[
|o(h)|\le h^{3/2}\hat Ku^2_nb^2_n,
\]
and $\hat K>0$ does not depend on $N$.
\end{theorem}

\begin{proof}
Fix $n=0, 1, \dots, N-1$ and define
\[
\UUU_{n+1}:=1+hu_na_n+\sqrt{h}u_nb_n\xi_{n+1}, \quad \nu_{n+1}:=hu_na_n+\sqrt{h}u_nb_n\xi_{n+1},
\]
and, for all $v\in \mathbb R$,
\[
\UU(v):=1+hu_na_n+\sqrt{h}u_nb_nv.
\]
Expand $\phi_\delta(\UUU_{n+1})$ by Taylor's formula and apply mathematical expectation \begin{equation*}
\begin{split}
\mathbf E \phi_\delta(\UUU_{n+1})=\phi_\delta(1)+ \phi'_\delta(1)\mathbf E \nu_{n+1}+\frac{\phi_\delta''(1)}2\mathbf E \nu_{n+1}^2+\mathbf E\left[\frac{\phi_\delta'''(\theta) }{6} \nu_{n+1}^3\right],
\end{split}
\end{equation*}
where $\theta$ is situated  between $1$ and $1+hu_na_n+\sqrt{h}u_nb_n\xi_{n+1}$.
Applying \eqref{bounds:ab} we arrive at the estimate
\begin{equation*}
\begin{split}
\left|\mathbf E\left[\ \frac{\phi_\delta'''(\theta) }{6}\nu_{n+1}^3 \right]  \right|\le & K_1\mathbf E\left|hu_na_n+\sqrt{h}u_nb_n\xi_{n+1}\right|^3\\&\le  K_2|u_n|^3h^{3/2}[ha_n^3+3a_nb_n^2]\le K_3u_n^2b_n^2h^{3/2},
\end{split}
\end{equation*}
where $K_i$, $i=1, 2, 3$,  does not depend on $n$. Note also that
\[
\phi_\delta(1)=\phi(1) , \quad \phi'_\delta(1)=\phi'(1), \quad \phi''_\delta(1)=\phi''(1).
\]
So the only thing which needs to be done  is to estimate
\[
\Delta_2:=\mathbf E\left|\phi(\UUU_{n+1})-\phi_\delta(\UUU_{n+1})\right|=\frac 1{\sqrt{2\pi}}\int_{|(\UU(v)|\le \delta}\left|\phi(\UU(v))-\phi_\delta(\UU(v))\right|e^{-v^2/2}dv.
\]
Change the variables by
\[
s=1+hu_na_n+\sqrt{h}u_nb_nv, \quad v=\frac{s-1-hu_na_n}{\sqrt{h}u_nb_n}, \quad dv=\frac {ds}{\sqrt{h}u_nb_n}.
\]
Assume that $\delta$ and $h_0>0$ are small enough and $|s|\le \delta$,  $h\le h_0$. Then,
for $u_nb_n>0$ we have
\[
 v=\frac{s-1-hu_na_n}{\sqrt{h}u_nb_n}\le \frac{\delta-1-hu_na_n}{\sqrt{h}u_nb_n}\le -\frac 1{2\sqrt{h} u_nb_n} \le-\frac 1{2\sqrt{h} \hat u\hat b},
  \]
  while $u_nb_n<0$ we have
\[
 v =  \frac{1-s-h|u_na_n|}{\sqrt{h}|u_nb_n|}\ge \frac{1-\delta-h|u_na_n|}{\sqrt{h}|u_nb_n|}\ge\frac 1{2\sqrt{h} |u_nb_n|} \ge \frac 1{2\sqrt{h} \hat u\hat b}.
  \]
  So
  \[
  |v|\ge \frac 1{2\sqrt{h}|u_n b_n|},
  \]
 which implies that
 \[
 e^{-v^2/2}\le  K_4 v^{-4}\le K_5h^2u_n^4 b_n^4.
 \]
 Note that $h_0>0$ chosen here does not depend on $n$, but only on bounds for $a$, $b$, $u$,  i.e. on $\hat a, \hat b, \hat u, \underline a, \underline b, \underline u>0$ (see  \eqref{bounds:ab}, \eqref{bound:pi}).

 This gives us
\begin{equation*}
\begin{split}
\Delta_2=&\frac 1{\sqrt{2\pi}}\int_{|(\UU(v)|\le \delta}\left|\phi(\UU(v))-\phi_\delta(\UU(v))\right|e^{-v^2/2}dv\\&\le\frac {K_5h^2u_n^4b_n^4}{\sqrt{2\pi}}\int_{|(\UU(v)|\le \delta}\left|\phi(\UU(v))-\phi_\delta(\UU(v))\right|dv\\&=\frac {K_5h^2u_n^4b_n^4}{\sqrt{2\pi}\sqrt{h}u_nb_n}\int_{|s|\le \delta}\left|\phi(s)-\phi_\delta(s)\right|ds\le K_6h^{3/2}u_n^2b_n^2,
\end{split}
\end{equation*}
where $K_6$ does not depend on $n$, which  completes the proof.

\end{proof}

\begin{corollary}
\label{cor:alpha}
For $\phi$, defined by \eqref{def:phi}, formula \eqref{cal:DIF} takes the form
\begin{equation}
\label{cal:DIFalpha}
\begin{split}
&\mathbf E\left(\phi(1+hu_na_n+\sqrt{h}u_nb_n\xi_{n+1})\right)\\
&=1+ h\alpha u_na_n+h\frac{\alpha(\alpha-1)}2u^2_nb^2_n+o(h),
\end{split}
\end{equation}
where
\[
|o(h)|\le h^{3/2}\hat Ku^2_nb^2_n,
\]
and $\hat K>0$ does not depend on $n=0, 1, \dots, N-1$. Hence  \eqref{cal:DIFalpha} can be written as
\begin{equation}
\label{cal:DIFalpha1}
\begin{split}
&\mathbf E\left(\phi(1+hu_na_n+\sqrt{h}u_nb_n\xi_{n+1})\right)\\
&=1+ h\alpha u_na_n+h\frac{\alpha(\alpha-1)}2u^2_nb^2_n[1+h^{1/2}O_n(1)],
\end{split}
\end{equation}\uu
where $|O_n(1)|\le \hat K$ for all $n=0, 1, \dots, N-1$, $N>\frac T{h_0}$ and $h\le h_0$.
\end{corollary}


\subsection{Positivity of $x_n$ with probability close to 1}
\label{sec:pos}

In this section we follow ideas from \citet{AGRK}, showing that  even though a.s. positivity is impossible to achieve for solution of \eqref{eq:discrw}, positivity with arbitrarily high probability is observed as the number of mesh points $N$  grows large. Again, we are giving the sketch of the proof instead of adapting  a result from  \citet{AGRK}.

Let  $x_n$ be a solution to \eqref{eq:discrw} with a positive initial value $x_0>0$, a parameter $h=\frac TN$ and  a strategy $\uu$.
Define
\begin{equation}
\label{def:OmegaN}
\Omega_N:=\mathbb P\{\omega\in \Omega: x_n(\omega)> 0, \quad  n=1, \dots, N\}.
\end{equation}

\begin{lemma}
\label{lem:pos}
Let Assumptions \ref{as:ab}, \ref{as:xi}  and condition \eqref{bound:pi} hold.
Let $\Omega_N$ be defined as in \eqref{def:OmegaN}. Then,  for each $\gamma\in (0, 1)$, we can find $N(\gamma)$ such that for all $N\ge N(\gamma)$
\[
\mathbb P[\Omega_N]\ge 1-\gamma.
\]
\end{lemma}
\begin{proof}
Note that  $x_n$ is $\mathcal F_n$-measurable and is independent of $\xi_{n+1}$. Let $u_nb_n>0$. Then, for $n=0, 1, \dots, N-1$,  we have
\begin{equation}
\begin{split}
\label{calc:pos1}
&\mathbb P\{x_{n+1}> 0\bigl |x_n>0\}=\mathbb P\left\{x_n\left(1+hu_na_n+\sqrt{h}u_nb_n\xi_{n+1}\right)> 0\bigl |x_n>0\right\}\\&=\mathbb P\left\{1+hu_na_n+\sqrt{h}u_nb_n\xi_{n+1}> 0\bigl |x_n>0\right\}=
\mathbb P\left\{\xi_{n+1}> -\frac{1+hu_na_n}{\sqrt{h}u_nb_n}\biggl |x_n>0\right\}\\&=1-\Phi\left(-\frac{1+hu_na_n}{\sqrt{h}u_nb_n}\right)=\Phi\left(\frac{1+hu_na_n}{\sqrt{h}u_nb_n}\right),
\end{split}
\end{equation}
where $\Phi$ is a normal probability distribution function.

If $u_nb_n<0$ we consider $\bar \xi_{n+1}=-\xi_{n+1}$ and note that $\bar \xi_{n+1}$ is also standard normal variable. So calculations \eqref{calc:pos1} holds true in this case again.

Applying  \eqref{calc:pos1},  Mill's  estimate (see \citet{KarShr})
\begin{equation*}
\frac{x}{(1+x^2)\sqrt{2\pi}}e^{-x^2/2}\le 1-\Phi(x)\le \frac1{x\sqrt{2\pi}}e^{-x^2/2}, \quad x>0,
\end{equation*}
and the inequality
\[
\frac{1+hu_na_n}{\sqrt{h}u_nb_n}\ge \frac{1}{2\sqrt{h}\hat u\hat b},
\]
we conclude that for some $h_1>0$ and all $h<h_1$, we have
\begin{equation*}
\begin{split}
&\mathbb P\{x_{n+1}> 0\bigl |x_n>0\}\ge \Phi\left(\frac{1+hu_na_n}{\sqrt{h}u_nb_n}\right)\ge 1-\frac{1}{\sqrt{2\pi}}\frac{e^{-\frac 12 \left(\frac{1+hu_na_n}{\sqrt{h}u_nb_n}\right)^2}}{
\frac{1+hu_na_n}{\sqrt{h}u_nb_n}}\\&
\ge  1-K_1\frac{\left(\frac1{2\sqrt{h}\hat u\hat  b} \right)^{-3}}{\frac1{2\sqrt{h}\hat u\hat b}}=
1-K_1\left({2\sqrt{h}\hat u\hat b} \right)^4=1-K_2h^2,
\end{split}
\end{equation*}
where $K_1, K_2>0$ do not depend on $n$. Then,
\begin{equation*}
\begin{split}
\mathbb P[\Omega_N]:=&\prod_{n=0, 1, \dots N-1}\mathbb P\{x_{n+1}> 0\bigl |x_n>0\}
\ge
\prod_{n=0, 1, \dots N-1}\left(1-K_2h^2\right)\\=
&\left(1-K_2h^2\right)^N=\left(1-\frac{K_2T^2}{N^2}\right)^N.
 \end{split}
\end{equation*}
Fix now $\gamma\in (0, 1)$ and find $N(\gamma)$ such that for all $N\ge N(\gamma)$
\[
1-\left(1-\frac{K_2T^2}{N^2}\right)^N<\gamma.
\]
This implies that  for all $N\ge N(\gamma)$
\[
\mathbb P[\Omega_N]\ge 1-\gamma,
\]
which completes the proof.
\end{proof}


\subsection{Estimation of maximum $\mathbf E\phi(x_N)$. }
\label{sec:max}

Let $h_0$, $O_n(1)$ and $\hat K$ be from Corollary \ref{cor:alpha}. So formula \eqref{cal:DIFalpha1} holds and  $|O_n(1)|<\hat K$ for all $n=1, \dots, N-1$,
$N>N_0=\frac T{h_0}$.
In addition we assume that $h_0$ is so small that  for $h\le h_0$, $n=1, \dots, N-1$, $N>N_0$,
\begin{equation}
\label{ineq:120}
1-h^{1/2}\hat K>0, \quad 1+h\alpha a_nu_n+h\frac{\alpha(\alpha-1)}2u_n^2b^2_n[1-\hat Kh^{1/2}]>0.
\end{equation}
Define, for $h\le h_0$,  $N>N_0$ and for any admissible strategy $\uu$,
\begin{equation}
\label{calc:prodopt}
\begin{split}
G(\uu):=\mathbf E\phi(x_N)&=\phi(x_0)\prod_{n=0}^{N-1}\mathbf E\left(\phi(1+hu_na_n+\sqrt{h}u_nb_n\xi_{n+1})\right)\\&=
x_0^\alpha\prod_{n=0}^{N-1}\left[1+\alpha h u_n a_n+h\frac{\alpha(\alpha-1)}2u_n^2b^2_n[1+h^{1/2}O_n(1)] \right],
\end{split}
\end{equation}
and
\begin{equation}
\label{def:mathcalG}
\hspace{-5cm}\mathcal {G}(\uu):=x_0^\alpha\prod_{n=0}^{N-1}\left[1+\alpha h u_n a_n+h\frac{\alpha(\alpha-1)}2u_n^2b^2_n\right].
\end{equation}


\subsubsection{Calculation of strategy  to maximize $\mathcal {G}(\uu) $.   }
\label{subsec:calc}

\begin{lemma}
\label{lem:maxphi}
Let Assumptions \ref{as:ab} and condition \ref{bound:pi}  hold. Let $\mathcal G$ be defined as in \eqref{def:mathcalG}. Then
the strategy  $\uu^*$,  defined by \eqref{def:pin}, maximizes  $\mathcal {G}(\uu)$.
\end{lemma}

\begin{proof}
To find the maximum of $\mathcal G$ we calculate its partial derivatives. We have
\begin{equation*}
\begin{split}
\frac{\partial \mathcal G}{\partial u_k}=&x_0^2\alpha h\left(a_k+(\alpha-1)u_kb^2_k\right)\prod_{n=0,
n\neq k}^{N-1}\left[1+h\alpha a_nu_n+h\frac{\alpha(\alpha-1)}2u_n^2b^2_n\right].
\end{split}
\end{equation*}
By \eqref{ineq:120}, solving the system
\begin{equation*}
\frac{\partial \mathcal G}{\partial u_k}=0, \quad  k=0, 1, \dots, N-1,
\end{equation*}
is equivalent to solving the system
\begin{equation}
\label{syst:max1}
a_k+(\alpha-1)u_kb^2_k  =0, \quad  k=0, 1, \dots, N-1.
\end{equation}
Solution $\uu ^*$ to \eqref{syst:max1} is given by \eqref{def:pin}. To show that
$\uu ^*$ is a point of maximum for the function $\mathcal {G}$, we find second partial derivatives of $\mathcal {G}$ at $\uu^*$. We have, for  $ k=0, \dots, N-1,$
\begin{equation*}
\begin{split}
\frac{\partial^2 \mathcal G}{\partial u_k^2}=&x_0^2\alpha h (\alpha-1)b^2_n\prod_{n=0, n\neq k}^{N-1}\left[1+\alpha h a_nu_n+h\frac{\alpha(\alpha-1)}2u_n^2b^2_n\right],
\end{split}
\end{equation*}
and, for $k\neq j$,
\begin{equation*}
\begin{split}
\frac{\partial^2 \mathcal G}{\partial u_k\partial u_j}=&x_0^2\alpha^2 h^2\left(a_k+(\alpha-1)u_kb^2_k\right)\left(a_j+(\alpha-1)u_jb^2_j\right)\\&\hspace{3cm}\times\prod_{n=0, n\neq k, j}^{N-1}\left[1+\alpha h  a_nu_n+h\frac{\alpha(\alpha-1)}2u_n^2b^2_n\right].
\end{split}
\end{equation*}
Let $y=(y_0, \dots, y_{N-1})$. Consider  the following quadratic form
\begin{equation}
\label{quadr:1}
\begin{split}
Q(y)=\sum_{k, j=0}^{N-1}\frac{\partial^2 \mathcal G}{\partial u_k\partial u_j}\biggl|_{\uu=\uu^*}y_ky_j.
\end{split}
\end{equation}
Since
\[
\left(a_k+(\alpha-1)u^*_kb^2_k\right)\left(a_j+(\alpha-1)u^*_jb^2_j\right)=0,
\]
we have,  for $k\neq j$,
\[
\frac{\partial^2 \mathcal G}{\partial u_k\partial u_j}\biggl|_{\uu=\uu^*}=0,
\]
and \eqref{quadr:1} takes the  form
\begin{equation*}
\begin{split}
&Q(y)=\sum_{k=0}^{N-1}\frac{\partial^2 G}{\partial u_k^2}\biggl|_{\uu=\uu^*}y_k^2
=x_0^2\alpha h(\alpha-1)\sum_{k}^Nb^2_k\prod_{n=0, n\neq k}^{N-1}\left[1+\frac{\alpha ha_n^2}{2(1-\alpha)b^2_n}\right]y_k^2.
\end{split}
\end{equation*}
Since $\alpha-1<0$, but
\[
\alpha h>0, \quad b^2_k>0, \quad 1+\frac{\alpha ha_n^2}{2(1-\alpha)b^2_n}>0,
\]
the quadratic form $Q(y)$ is negatively defined, which proves that $\uu^*$ given  by
\eqref{def:pin}  is a point of maximum for $\mathcal G$.

\end{proof}

\subsubsection{Estimation of the difference $G(\uu)-\mathcal G(\uu)$.}
\label{subsec:estdif}

\begin{lemma}
\label{lem:difGG}
Let Assumptions \ref{as:ab}, \ref{as:xi}  and condition \eqref{bound:pi} hold.
Let  $\mathcal{G}(\uu)$ and $G(\uu)$ be defined as in \eqref{def:mathcalG} and  \eqref{calc:prodopt}, respectively.
Then, for each $\varepsilon>0$ there exists $N(\varepsilon)\in \mathbf N$ such that for all $N>N(\varepsilon)$,
$h\le \frac T{N(\varepsilon)}$, we have
\begin{equation}
\label{ineqGG*}
|G(\uu)-\mathcal{G}(\uu)|\le \varepsilon.
\end{equation}
\end{lemma}

\begin{proof}
Denote
\begin{equation}
\label{def:3def}
\nu_n:=1+\alpha h u_n a_n+h\frac{\alpha(\alpha-1)}2u_n^2b^2_n, \quad \eta_n:=\frac{\alpha(\alpha-1)}2\frac{u_n^2b^2_n}{\nu_n}.
\end{equation}

 Let $h_0$ and $\hat K$ be as in Corollary \ref{cor:alpha} .
Assume in addition  that $h_0$ is so small  that
\begin{equation}
\label{cond:h02}
1+\alpha h u_n a_n+h\frac{\alpha(\alpha-1)}2u_n^2b^2_n[1-h^{1/2}\hat K]>\frac 12.
\end{equation}
Based on Assumption \ref{as:ab}, inequality \eqref{bound:pi} and \eqref{cond:h02}, we conclude that
there exist constants $\hat K_1>0$ and $\hat K_2>0$ which do not depend on $N$($h_0$) such that, for all $n=1, 2, \dots, N-1$,
\begin{equation}
\label{est:K12}
\hat K\eta_n\le \hat K_1, \quad \hat K\alpha(\alpha-1)u_n^2b^2_n\le \hat K_2.
\end{equation}
Note that $|O_n(1)|\le \hat K$ . Then, applying \eqref{cond:h02} and \eqref{est:K12} we have, for all $n=1, 2, \dots, N-1$,
\begin{equation}
\label{est:right1}
\begin{split}
&1+\alpha h u_n a_n+h\frac{\alpha(\alpha-1)}2u_n^2b^2_n[1+h^{1/2}O_n(1)]\le \nu_n+h^{3/2}\hat K\frac{\alpha(\alpha-1)}2u_n^2b^2_n=\\&\\
&\nu_n\left[ 1+  h^{3/2}\hat K \eta_n\right]\le \nu_ne^{ h^{3/2}\hat K \eta_n}\le
\nu_ne^{ h^{3/2}\hat K_1},
\end{split}
\end{equation}
 and
\begin{equation}
\label{est:left1}
\begin{split}
&1+\alpha h u_n a_n+h\frac{\alpha(\alpha-1)}2u_n^2b^2_n[1+h^{1/2}O_n(1)]\ge \nu_n-h^{3/2}\hat K\frac{\alpha(\alpha-1)}2u_n^2b^2_n=\\&\\
&\nu_n\left[\frac{\nu_n-h^{3/2}\hat K\frac{\alpha(\alpha-1)}2u_n^2b^2_n}{\nu_n}\right]
=\nu_n\left[\frac{\nu_n}{\nu_n-h^{3/2}\hat K\frac{\alpha(\alpha-1)}2u_n^2b^2_n}\right]^{-1}=\\&\\
&\nu_n\left[1+\frac{h^{3/2}\hat K\frac{\alpha(\alpha-1)}2u_n^2b^2_n}{\nu_n-h^{3/2}\hat K\frac{\alpha(\alpha-1)}2u_n^2b^2_n}\right]^{-1}\ge
 \nu_n
 \exp\left\{-\frac{h^{3/2}\hat K\frac{\alpha(\alpha-1)}2u_n^2b^2_n}{\nu_n-h^{3/2}\hat K\frac{\alpha(\alpha-1)}2u_n^2b^2_n} \right\}\ge \\&\\
 &\nu_n \exp\left\{-h^{3/2}\hat K\alpha(\alpha-1)u_n^2b^2_n\right\}\ge\nu_ne^{ -h^{3/2}\hat K_2}.
\end{split}
\end{equation}
Note that
\begin{equation}
\label{note:1}
\mathcal G(\uu)=x_0^\alpha\prod_{n=0}^{N-1}\nu_n, \quad h^{3/2}N=Th^{1/2}, \quad hN=T.
\end{equation}
Applying \eqref{est:left1}, \eqref{est:right1}  and \eqref{note:1} we arrive at
\begin{equation*}
\label{est:rightleft1}
\begin{split}
\mathcal G(\uu)e^{-h^{3/2}\hat K_2N}\le G(\uu)\le \mathcal G(\uu)e^{h^{3/2}\hat K_1N},
\end{split}
\end{equation*}
or
\begin{equation*}
\label{est:rightleft2}
\begin{split}
\mathcal G(\uu)e^{-h^{1/2}\hat K_2T}\le G(\uu)\le \mathcal G(\uu)e^{h^{1/2}\hat K_1T},
\end{split}
\end{equation*}
which implies that
\begin{equation}
\label{est:rightleft3}
\begin{split}
\mathcal G(\uu)\left[e^{-h^{1/2}\hat K_2T}-1\right]\le G(\uu)-\mathcal G(\uu)\le \mathcal G(\uu)\left[e^{h^{1/2}\hat K_1T}-1\right].
\end{split}
\end{equation}
Therefore,
\[
| G(\uu)-\mathcal G(\uu) |\le  G(\uu)\max\left\{1-e^{-h^{1/2}\hat K_2T}, \,\,   e^{h^{1/2}\hat K_1T}-1\right\}.
\]
Now we estimate $G(\uu)$. By Assumption \ref{as:ab} and inequality \eqref{bound:pi}
we have
\begin{equation}
\label{est:rmathcalG}
\begin{split}
\mathcal G(\uu)=&x_0^\alpha\prod_{n=0}^{N-1}\left[1+\alpha h u_n a_n-h\frac{\alpha(1-\alpha)}2u_n^2b^2_n\right]\le \\&x_0^\alpha\exp\left\{\alpha h\sum_{n=0}^{N-1}\left[u_n a_n-\frac{(1-\alpha)}2u_n^2b^2_n\right]\right\}\le \\&
x_0^\alpha\exp\left\{\alpha h N\left[\hat a\hat u-\frac{(1-\alpha)}2\underline u^2\underline b^2_n\right]\right\}=\\&
x_0^\alpha\exp\left\{\alpha T\left[\hat a\hat u-\frac{(1-\alpha)}2\underline u^2\underline b^2_n\right]\right\}=x_0^\alpha C_1,
\end{split}
\end{equation}
for some $C_1>0$, which does not depend on $N$ or $h$.

Now, fix $\varepsilon>0$ and find $N=N(\varepsilon)$ such that, for $h<\frac T{N(\varepsilon)}$,
\[
\max\left\{e^{h^{1/2}2\hat K_2}-1, \,\,  1-e^{-h^{1/2}2\hat K_1}\   \right\}<\frac{\varepsilon}{C_1}.
\]
Then, for $N>N(\varepsilon)$, inequality \eqref{ineqGG*} holds.
\end{proof}
\subsection{Estimation of $\max \mathbf E \psii(x_N)$.}
\label{subsec:maxpsi}


\subsubsection{Estimation of $\mathbf E |x_N|^2$}
From \eqref{pres:0} we obtain, for $n=1, 2 ,\dots, N$,:
\begin{equation*}
\begin{split}
&\mathbf E|x_n|^2=|x_0|^2\prod_{i=1}^{n-1}\mathbf E\left|1+hu_ia_i+\sqrt{h}u_ib_i\xi_{i+1}\right|^2\\=
&x_0^2\prod_{i=1}^{n-1}\mathbf E\left[1+h(2u_ia_i+hu^2_ia^2_i+ u^2_ib^2_i)+2\sqrt{h}(1+hu_ia_i)u_ib_i\xi_{i+1}+hu^2_ib^2_i(\xi_{i+1}^2-1)\right]\\=
&x_0^2\prod_{i=1}^{n-1}\left[1+h(2u_ia_i+hu^2_ia^2_i+ u^2_ib^2_i)\right]\le
x_0^2\prod_{i=1}^{n-1}\left[1+hK_3\right]\le |x_0|^2\left[1+hK_3\right]^n,
\end{split}
\end{equation*}
so
\begin{equation}
\label{est:E2}
\mathbf E|x_N|^2\le x_0^2\left[1+hK_3\right]^N=|x_0|^2e^{NhK_3}= x_0^2e^{K_3T}.
\end{equation}

\subsubsection{Estimation of $\max\mathbf E\psii(x_N)$.}
\label{subsec:ito}
Substituting  the value $\uu ^*$ from \eqref{def:pin} into  \eqref{def:mathcalG} we get
\begin{equation}
\label{calc:Gn*}
\begin{split}
\mathcal{G}(\uu ^*)=
x_0^\alpha\prod_{n=0}^{N-1}\left[1+\alpha
h\frac{a^2_n}{2(1-\alpha)b^2_n}\right].
\end{split}
\end{equation}

\begin{lemma}
\label{lem:maxpsi}
Let Assumptions \ref{as:ab}, \ref{as:xi}  and condition \eqref{bound:pi}  hold. Let  $x_n$ be a solution to \eqref{eq:discrw} with a positive initial value $x_0>0$, a parameter $h=\frac TN$ and  a strategy $\uu$. Let  $\mathcal{G}(\uu^*)$ be  defined  as in \eqref{calc:Gn*} and $\psii$ be  defined as in \eqref{def:psi}. Then, for each $\varepsilon>0$, there exists $N(\varepsilon)\in \mathbf N$ such that for all $N>N(\varepsilon)$, $h\le \frac T{N(\varepsilon)}$, we have
\begin{equation}
\label{ineqGGG*}
|\sup_{u} \mathbf E \psii(x_N)-\mathcal{G}(\uu ^*)|\le \varepsilon.
\end{equation}
\end{lemma}

\begin{proof}
Fix $\gamma\in (0, 1)$ and find $N(\gamma)$ by Lemma \ref{lem:pos}. Then, by definition \eqref{def:psi} of $\psii$, for $\Omega _N$, defined by \eqref{def:OmegaN} with $N\ge N(\gamma)$,  we have
\[
\psii(x_N(\omega))=\phi(x_N(\omega))=|x_N(\omega)|^\alpha, \quad \omega\in \Omega_N,
\]
so
\[
\mathbf P\{\omega:  \psii(x_N(\omega))\neq \phi(x_N(\omega))\}\le \mathbb P[\Omega\setminus\Omega_N]\le \gamma.
\]
Further,
\begin{equation}
\label{est:dif2}
\begin{split}
&\hspace{-2.3cm}\mathbf E\left|\phi(x_N)- \psii(x_N)\right|=\int_{\Omega}\left|\phi(x_N(\omega))- \psii(x_N(\omega))\right|dP(\omega)\le\\&\int_{\Omega\setminus \Omega_N}\left[\left|x_N(\omega)\right|^\alpha+L\left|x_N(\omega)\right| \right]dP(\omega)=\\&\int_{\Omega\setminus \Omega_N}\left|x_N(\omega)\right|^\alpha dP(\omega)+
L\int_{\Omega\setminus \Omega_N}\left|x_N(\omega)\right| dP(\omega)\le
\\&\left(\int_{\Omega\setminus \Omega_N}\left|x_N(\omega)\right|^{2} dP(\omega)\right)^{\frac{\alpha}2}\times \left(\int_{\Omega\setminus \Omega_N}dP(\omega)\right)^{\frac{2-\alpha}2}+\\&
L\left(\int_{\Omega\setminus \Omega_N}\left|x_N(\omega)\right|^{2} dP(\omega)\right)^{\frac 12}\times \left(\int_{\Omega\setminus \Omega_N}dP(\omega)\right)^{\frac 12}\le\\&
\left(\int_{\Omega}\left|x_N(\omega)\right|^{2} dP(\omega)\right)^{\frac{\alpha}2}\left(\mathbf P\{\Omega\setminus \Omega_N\}\right)^{\frac{2-\alpha}2}+\\&L\left(\int_{\Omega}\left|x_N(\omega)\right|^{2} dP(\omega)\right)^{\frac 12}\left(\mathbf P\{\Omega\setminus \Omega_N\}\right)^{\frac 12}\le\\&
\hspace{4.5cm}\left(\mathbf E|x_N|^2\right)^{\frac{\alpha}2}\gamma^{\frac{2-\alpha}2}+L\left(\mathbf E|x_N|^2\right)^{\frac 12}\gamma^{\frac 12}.
\end{split}
\end{equation}
Since $1-\frac \alpha 2>\frac 12$ and $\gamma\in (0, 1)$ estimates \eqref{est:E2} and \eqref{est:dif2} imply
\[
\mathbf E\left|\phi(x_N)- \psii(x_N)\right|\le K_4\gamma^{\frac12},
\]
where $K_4>0$ does not depend on $N$.

Then
\[
 \left|\mathbf E \psii(x_N)-\mathbf E \phi(x_N)\right|\le \mathbf E\left|\phi(x_N)- \psii(x_N)\right|\le K_4 \gamma^{\frac 12}.
\]
Now, fix  $\varepsilon>0$, and choose
\begin{equation}
\label{def:gamma}
\gamma<\left(\frac{\varepsilon}{2K_4}\right)^{2}.
\end{equation}
By  Lemma \ref{lem:pos},  find $N(\gamma)$.  By Lemma \ref{lem:difGG}, find $N(\varepsilon/2)\ge N(\gamma)$ such that, for each admissible strategy $\uu$, $N\ge  N(\varepsilon/2)$  and $h\le \frac T{N(\varepsilon/2)}$ we have
\begin{equation}
\label{est:varep2}
|G(\uu)-\mathcal{G}(\uu )|\le \varepsilon/2.
\end{equation}
Recall that $\mathbf E \phi(x_N)=G(\uu )$ and $\sup_{u}\left[\mathcal G(\uu)\right]=\mathcal G(\uu^*)$. Then, by \eqref{def:gamma} and \eqref{est:varep2}   we have, for  $N\ge  N(\varepsilon/2)$ and $h\le \frac T{N(\varepsilon/2)}$,
\begin{equation}
\label{est:maximum}
\begin{split}
&|\sup_{u} \mathbf E \psii(x_N)-\mathcal G(\uu ^*)|=|\sup_{u}\left[ \mathbf E \psii(x_N) -\mathbf E \phi(x_N)+G(\uu )-\mathcal G(\uu)+\mathcal G(\uu)\right]-\mathcal G(\uu ^*)|\le
\\&|\sup_{u}\left[ \mathbf E \psii(x_N) - \mathbf E \phi(x_N)\right]|+|\sup_{u}\left[G(\uu )-\mathcal G(\uu)\right]|+|\sup_{u}\left[\mathcal G(\uu)\right]-\mathcal G(\uu^*)|\le
\\&
\le K_4\gamma^{\frac 12}+\frac \varepsilon 2\le \varepsilon.
\end{split}
\end{equation}
\end{proof}


Now we are able  to complete the proof of Theorem \ref{ThM}.
For small enough $h=\frac TN$,  the  strategy $u^*$ defined by \eqref{def:pin}  maximize $\mathbf E \psii(x_N)$ approximately, meaning that
\[
\sup_{u} \mathbf E \psii(x_N)=x_0^2\prod_{n=0}^{N-1}\left[1+\alpha h\frac{a^2_n}{2(1-\alpha)b^2_n}\right]+\rho(N),
\]
where $\rho(N)\to 0$ as $N\to \infty$. Then the proof of Theorem \ref{ThM} follows.
\bibliographystyle{ieeetr}

\section*{Simulations }
The equation of type \eqref{eq:discrw} with $u_n\equiv 1$, $a_n\equiv \lambda$, $b_n\equiv \mu$ was considered in  \citet{Palm}, where an explicit bound on $h$ suitable for application of  the discrete It\^o formula was computed and error terms were estimated.

\citet{AGRK} developed an  asymptotic estimate on the number of mesh points $N(\gamma)$ required to ensure the positivity of solutions of Euler-Maruyama  discretization
of \eqref{eq:contw} as the required proportion of positive trajectories $\gamma$ approaches 1. Based on the above works one can estimate $h$ (or $N$) for equation \eqref{eq:discrw}  in order to be able to apply the discrete It\^o formula and to ensure  positivity with given probability $\gamma$.

For the purposes of the present paper, it suffices to demonstrate the impact  of the sampling interval of discretization on
the performance of Merton's strategy. In addition, we want to show the impact  of the selection  of finite $L$ in the adjusted utility function.
We remind that the classical concave utility function corresponds to  the case $L=+\infty$ that we excluded.

Table \ref{tab1} presents  the results of the numerical simulations. This table shows the differences
 $E U(X^*_T) -E U(x_N^*)$,  where $E U(x^*_N)$ is the expected utility for the strategy described in  Theorem \ref{ThM}, and where
   $$
   E U(X^*_T)=E {X^*_T}^\alpha=\exp\left(\frac{a}{2b^2}\frac{\alpha}{1-\alpha}\right)
    $$
    is the expected utility for Merton's strategy in the continuous time setting. Note that the value $X^*_T$ is non-negative, and therefore is not impacted by the choice of $L$.

   The table shows the values of these differences for
  the parameters   $L=10,10^2,10^3,10^5,10^6$ presented in (\ref{def:psi}) and for  $N=2,6,12,52,250$, with $X_0=1$, $T=1$, $\alpha=1/2$, $T=1$,
 $a = 0.07$, $b = 0.2$, with 1,000,000 simulations each. For these parameters,   $
   E U(X^*_T)=1.0631$.

The values  $N$  show the numbers of allowed portfolio adjustments during one year time period, with $\delta=T/N=1/N$.

It can be seen that $n=52$ (i.e, $\delta=1/52$) that corresponds to weekly portfolio adjustments is sufficient to compensate the discretization error for Merton's strategy.
This error is almost negligible for  $n=250$ (i.e, $\delta=1/250$) that corresponds to daily portfolio adjustments.
\begin{table}[H]\label{tab1}
\center
\begin{tabular}{|l||*{6}{c|}}\hline
\backslashbox{L}{N}
&\makebox[3em]{2}&\makebox[3em]{6}&\makebox[3em]{12}
&\makebox[3em]{52}&\makebox[3em]{250}\\\hline\hline
$10$      &0.054938	& 0.004632	& 0.001674	& 0.000803	& $9.506101 \times 10^{-5}$	\\\hline
$10^2$    &0.462369	& 0.010440	& 0.001706	& 0.000803	& $9.506101 \times 10^{-5}$ \\\hline
$10^3$    &4.536681	& 0.068490	& 0.002021	& 0.000803  & $9.506101 \times 10^{-5}$	\\\hline
$10^5$    &452.7109	& 6.453968	& 0.036671	& 0.000803	& $9.506101 \times 10^{-5}$	\\\hline
$10^6$    &4527.022	& 64.50377	& 0.351671	& 0.000803	& $9.506101 \times 10^{-5}$	\\\hline
\end{tabular}
\caption{The differences $\sup_{N,u}E X_T^\alpha -E U(X_N)$ in Theorem \ref{ThM} for different values of $L$ in (\ref{def:psi}) and for different numbers $N$ of portfolio adjustments during one year time period. }
\end{table}

\section*{Conclusions }  We have investigated the possibility of using
known optimal continuous time strategies for solving
the discrete time optimal portfolio selection problems.
For this, we studied  the limit
properties of the discrete time optimal portfolio selection problem
when the step of the discretization converges to zero.  We found that
the  performance of the discrete time strategy obtained directly
from Merton's strategy approximates the optimal strategy after
some minor adjustment of the utility function.
This
suboptimal discrete time strategy is myopic.
  The proof is based  on the application of
a  discrete  It\^o formula.  The results  of this paper leads to the  conclusion that Merton's strategies can be  used effectively for discrete time multi-period market models.


  \subsubsection*{Acknowledgment} This work  was supported by ARC grant of Australia DP120100928.\\
The authors extend their appreciations to the anonymous referees for their valuable remarks {  which  helped us} to improve the paper.

\bibliographystyle{apalike}

 \end{document}